%% file: adv_paper.tex
\documentclass[11pt]{article}

\usepackage{fullpage, amsmath, amssymb, amsthm, enumerate, color}
\usepackage{nicefrac}
\usepackage{algorithm}
\usepackage{multicol}
\usepackage[noend]{algpseudocode}
\usepackage{graphicx, tikz, algorithmicx, natbib}
\usepackage{dsfont}

\newtheorem{theorem}{Theorem}[section]
\newtheorem{proposition}[theorem]{Proposition}
\newtheorem{lemma}[theorem]{Lemma}
\newtheorem{claim}[theorem]{Claim}
\newtheorem{corollary}[theorem]{Corollary}

\def\P{\mathbb{P}}
\def\bW{{\bf v}}

\def\E{\mathbb{E}}

\def\A {\mathcal{A}}

\def\argmax{\text{argmax}}

\begin{document}

\title{Maximizing Efficiency in Dynamic Matching Markets}

\author{Itai Ashlagi, Maximilien Burq, Patrick Jaillet, Amin Saberi.}

\date{}
\maketitle
\begin{abstract}

We study the problem of matching  agents who arrive at a marketplace over time and leave after $d$ time periods. Agents can only be matched while they are present in the marketplace. Each pair of agents can yield a different match value, and the planner's goal is to maximize the total value over a finite time horizon. We study matching algorithms that perform well over any sequence of arrivals when there is no a priori information about the match values or arrival times.

Our main contribution is a $\nicefrac{1}{4}$-competitive algorithm. The algorithm randomly selects a subset of agents who will wait until right before their departure to get matched, and maintains a maximum-weight matching with respect to the other agents. The primal-dual analysis of the algorithm hinges on a careful comparison between the initial dual value associated with an agent when it first arrives, and the final value after $d$ time steps.

It is also shown that no algorithm is  $\nicefrac{1}{2}$-competitive. We  extend the model to the case in which departure times are drawn i.i.d from a distribution with non-decreasing hazard rate, and establish a $\nicefrac{1}{8}$-competitive algorithm in this setting. Finally we show on real-world data that a modified version of our algorithm performs well in practice.
\end{abstract}

\input{intro.tex}
\input{model.tex}
\input{main_result.tex}
\input{analysis.tex}
\input{stochastic_dep.tex}
\input{hardness.tex}

\input{numerical.tex}
\input{conclusion.tex}
\bibliography{bibliography}
\bibliographystyle{apalike}

\appendix
\input{appendix.tex}

\end{document}

%% file: intro.tex

\section{Introduction}

We study  the problem of  matching agents who arrive to a marketplace over time and leave after a short period. Agents can only be matched while they are present in the marketplace. There is a different value for matching every pair of agents, which does not vary over time. The planner's goal is to maximize the total value over a given finite time horizon.

Several marketplaces face a such a problem. Ride-hailing platforms have to match passengers with drivers, in which case the value of a match may depend on to the distance between the driver and the passenger. Such platforms may also carpool passengers and hence match passengers with each other, in which case the value of a match can represent the reduction in total distance traveled by the two matched passengers, compared to the distance traveled in individual rides. Kidney exchange platforms face the problem of matching incompatible patient-donor pairs with each other. In this context the value of a match can represent, for example, the quality adjusted life years due to the transplant.
The common challenge in all these applications comes from the uncertainty associated with future arrivals and potential future matches.

We study matching algorithms that perform well across any sequence of arrivals, when there is no a priori  information about the match values or arrival times. The underlying graph structure may be arbitrary and is not necessarily bipartite. Agents can be matched at any moment between their arrival and their departure. In that sense, our framework differs from the classic online matching literature where matching decisions have to be made immediately upon the arrival of an agent.

One important assumption we make is that each agent departs from the market  exactly $d$ time periods after her arrival. In the case of the carpooling application one may think of $d$ as a (self-imposed) service requirement, which ensures that no passenger waits for too long before being matched. In that case, after $d$ periods the passenger is assigned to an individual ride. Later on, we relax this assumption to allow for stochastic departures. It is worth noting that when departure times are allowed to be arbitrary, the competitive ratio of any algorithm is unbounded.

\subsection*{Main contributions}

We introduce an algorithm, termed {\it Postponed Dynamic Deferred Acceptance} (PDDA), that achieves a competitive ratio of $\nicefrac{1}{4}$. We further show that no algorithm  achieves a competitive ratio that is higher than $\nicefrac{1}{2}$.

A key step of the algorithm is to artificially create a two-sided market by randomly assigning each agent to either be a {\it buyer} or a {\it seller}. In this market, buyers will ``bid'' to match with a desired seller. With each seller that arrives to the market we associate a price, which is initiated to zero. For each buyer, the marginal utility of matching with a given seller is the value from the match with that seller minus the price that the seller  demands. The algorithm maintains these virtual prices (for sellers) and profit margins (for buyers).

Once a buyer  joins the  market she triggers  a bidding process similar to an ascending auction \citep{kuhn1955hungarian,demange1986multi,bertsekas1988auction}.
This ascending bidding process maintains a tentative matching. A tentative match between a buyer and a seller is converted into a real match only if the seller has been present for $d$ time periods and is about to depart. At that time, both the seller and the buyer to whom she is matched depart from the market. Given that sellers are patient and choose their match in the last minute, a buyer will never bid on sellers who arrive after her. This carefully chosen bidding and matching process guarantees that the profits of the agents (sellers' prices and buyers' marginal utilities) are monotone over time, which is crucial for our primal-dual competitive analysis.

We now describe in more detail how the agents are assigned to become buyers or sellers. One simple way to do so is to flip an unbiased coin, independently for each agent, at the time of arrival. We improve upon this naive assignment by introducing two copies of each agent, and assign one to be a seller and the other to be a buyer. This enables us to postpone the decision until we have more information about the graph structure and the likely matchings.

We extend our model to the case in which departures are stochastic. We show that when the departure distribution has a non-decreasing hazard rate and departure times are known at the last minute, we can adapt our algorithm to achieve a competitive ratios of $\nicefrac{1}{8}$.

\subsection*{Related literature}

This work has ties to the online  matching  problem. In the classical problem proposed by \cite{kvv}, the graph is bipartite and a finite number of vertices on one side are waiting, while vertices on the other side arrive dynamically and have to be matched {\it immediately} upon arrival. This work has numerous extensions, for example to stochastic arrivals, and the adwords setting \cite{mehta2007adwords,aryanak_randominput,aryanak_stmatching,mos,STMatchingPatrick}. See \cite{mehta2013online} for a detailed survey.

Our work differs from the above line of work in three ways. First, we allow both sides of the graph to arrive and depart dynamically. This is useful for applications such as dynamic matching of drivers to passengers. Second, we are able to provide algorithms that perform well in the case of edge-weighted inputs. Lastly, we do not require the graph to be bipartite, which is useful in the case of kidney exchange, as well as dynamic matching of carpooling users.

This paper is also related to a presentation by \cite{dutta2017}, which looks at a model where the online algorithm uses advance knowledge of future arrivals. Closely related is  \cite{huang2018}, who study a similar model but  consider the unweighted case  and establish bounds also for adverserial departures.

A large literature has been dedicated to the   \emph{static} matching  with heterogeneous match values and a particular problem is finding a maximum-weight matching problem efficiently. Classic algorithms that have been proposed include the \emph{Hungarian} algorithm \citep{kuhn1955hungarian}, and \emph{auction} algorithms \cite{demange1986multi,bertsekas1988auction}. For the case, in which agents have ordinal preferences, \citet{gale1962college} proposed the \emph{Deferred Acceptance} algorithm which finds a \emph{stable} matching. 
Our work builds on these algorithms by maintaining a tentative maximum-weight matching over time, and matches are made final when a seller is about to depart.

There is also a growing  literature that focuses on dynamic matching in the context of  kidney exchange, initiated by \citet{Utku}. Most papers focus on random graphs, and study various questions such as: the effect of long cycles and chains \cite{anderson2015dynamic}, the effect of edge failure \cite{dickerson2013failure}, or the effect of a having vertex heterogeneity \cite{ashlagi2017matching}. Closer to out paper is work by \cite{akbarpour2017thickness} which studies the effect of knowing when a vertex is about to depart. They show that in a sparse random graph model, matching vertices only when they become  \textit{critical}  performs well. Our contribution with respect to these papers are the following. First we provide a framework to study dynamic matching without requiring a specific random graph model. Second, these papers focus on the non-weighted case and in contrast we consider the {\it weighted} case. Lastly, we show that in a worst-case setting, having information on \textit{critical} vertices is necessary but may not be sufficient to obtain an efficient matching algorithm: optimization is also needed.

Few papers have considered the problem of dynamic matching when matches yield different values. In the adversarial setting \citet{emek2016online,ashlagi2017min} study the problem of minimizing the sum of distances between matched agents and the sum of their  waiting times. Their model has no departing agents and everyone is matched and our model does not account for agents' waiting times.in the case where o
Several papers  consider the problem of dynamic matching in stochastic environments  \citep{baccara2015optimal,ozkan2016dynamic,hu2016dynamic} (the latter paper allows agents to depart the market). These papers find that some accumulation of agents  is beneficial for improving efficiency. More recently, \citep{truong2018} study matching when the graph is bipartite and vertices arrive stochastically. They provide a $\nicefrac{1}{4}$ approximation in the case where one side departs immediately while the other side departs after some arbitrary time.

%% file: model.tex

\section{Model}

We consider a model with a finite horizon, in which at each period $i=1\dots,T$ a single agent, also denoted $i$, arrives to the market. Note that for the simplicity of notation, we refer to the agent that arrives at time as agent $i$.
The value (weight) from matching agent $i$ with agent $j$ is denoted by $v_{i,j} \geq 0$.
Each agent in our model has to leave exactly $d$ time steps after arriving to the market.
Agents can only be matched when they are in the market and matched agents are removed from the market immediately.

We assume that the planner knows $d$ as well as the match value between any pair of agents $i$ and $j$ only if both $i$ and $j$ are present in the market. We say that a vertex is \emph{critical} when she has been in the market for $d$ time periods. In that case, the planner has the option to match her before she leaves. The planner has no distributional information about future arrivals or future  match values.

The planner's goal is to find a matching that will yield a maximum total value for any set of values $\bW$ and any order of arrivals. In  particular we are interested in  designing  a matching algorithm that achieves a high competitive ratio with respect to the maximum-weight matching at hindsight.

It will be convenient to represent the input as a {\it graph}, where each agent is represented by a vertex, and there is an undirected edge between vertices $i$ and $j$ if and only if $|i - j| \leq d$.

To illustrate a natural tradeoff, consider the example in  Figure \ref{fig:hard:basic}, in which every agent remains $d=1$ periods in the market. At period $2$ the planner can either match agents $1$ and $2$ or let agent $1$ leave the market unmatched. This simple example show that no deterministic algorithm can obtain a constant competitive ratio. Furthermore, no algorithm can achieve a competitive ratio higher than $\nicefrac{1}{2}$. See  Section \ref{sec:hardness} for a more detailed discussion.

\begin{figure}[H]
  \centering
  \begin{tikzpicture}[pre/.style={<-,shorten <=1.5pt,>=stealth,thick}, post/.style={->,shorten >=1pt,>=stealth,thick}, scale=0.7]
  \tikzstyle{every node}=[draw,shape=rectangle,minimum size=5mm, inner sep=0];
  \tikzstyle{edge} = [draw,thick,-]
  \tikzstyle{every node}=[shape=circle,minimum size=8mm, inner sep=0];

  \draw [fill](-2,0) circle [radius=0.2];
  \node [right] at (-2.7,-0.7) {$1$};
  \draw [fill](0,0) circle [radius=0.2];
  \node [right] at (-0.7,-0.7) {$2$};
  \draw [fill](2,0) circle [radius=0.2];
  \node [right] at (1.3,-0.7) {$3$};

  \draw[line width=1.5pt] [-] (-2+0.35, 0) .. controls(-1,0) ..(0-0.35, 0);
  \node [right] at (-1.7,0.7) {\small $v_{1,2} = 1$};
  \draw[line width=1.5pt] [-] (0+0.35, 0) .. controls(1,0) .. (2-0.35, 0);
  \node [right] at (0.3,0.7) {\small $v_{2,3} = y$};
  \end{tikzpicture}
  \caption{Let $d = 1$. Agent $1$ becomes critical before the arrival of agent $3$. Therefore, the planner needs to decide whether to match $1$ and $2$ and collect $v_{1,2}$ without knowing $y$.}
  \label{fig:hard:basic}
\end{figure}
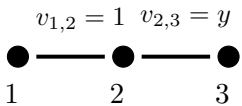

%% file: main_result.tex

\section{Main results}

The example in Figure \ref{fig:hard:basic} illustrates a necessary condition for the algorithm to achieve a constant competitive ratio: with some probability, vertex $2$ needs to forgo the match with vertex $1$ and wait until she becomes critical. In general, we ensure this property by assigning every agent to be either a {\it seller} or a {\it buyer}.  Buyers may get matched and leave the market at any time but sellers do not match before they become critical.

We will  first consider  the case in which the underlying graph is  bipartite, with buyers on one side and sellers on the other. To ensure that sellers never leave before they become critical, we further assume that there is no edge between a seller $s$ and any buyer $b$ who arrived before $s$. Such a graph is called a {\it constrained} bipartite graph.

In Section \ref{sec:bipartite} we introduce an algorithm for the case in which  the input graph is constrained bipartite. Through a primal-dual analysis,  we establish that the algorithm achieves at least \nicefrac{1}{2} of the offline total reward.

In Section \ref{sec:arbitrary} we consider  arbitrary graphs. For such graphs, we artificially generate a two-sided market by creating a buyer and a seller copy of each vertex. The matching produced by the algorithm in this  bipartite graph, is then transformed into a matching in the original graph by a carefully constructed randomized process. This process loses an additional factor of $2$.

\subsection{Algorithm for constrained bipartite graphs}
\label{sec:bipartite}

We assume here that the input is a constrained bipartite graph. Vertices on one side are called {\it buyers} and vertices on the other side are called  {\it sellers}, and there is no edge between a buyer $b$ and a seller $s$ if $b$ arrives before $s$ ($b<s$).

Note that even in this simplified setting no online algorithm can find the optimum solution. This is because when a new buyer arrives, she may create an augmenting path in the offline graph in which some of the vertices are already departed. We describe this in more details in  Claim \ref{cl:ex:bipartite_constrained}.

We  introduce the Dynamic Deferred Acceptance (DDA) algorithm, which takes as input a constrained bipartite graph and returns a matching.
The main idea is to maintain a temporary maximum-weight matching $m$ at all times during the run of the algorithm. This matching is updated according to an auction mechanism: every seller $s$ is associated with a \emph{price} $p_s$, which is initiated at zero upon arrival. Every  buyer $b$ that is present in the market is associated with a \emph{profit margin} $q_b$ which corresponds to the value of matching to their most prefered seller minus the price associated with that seller.
Every time a new buyer joins the market, she bids on her most prefered seller at the current set of prices. This triggers a bidding process that terminates when no unmatched buyer can profitably bid on a seller.

A tentative match between a buyer and a seller is converted into a real match only if the seller is critical, i.e. she has been present for $d$ time periods and is about to depart. At that time, both the seller and the buyer to whom she is matched depart from the market. This ensures that sellers never get matched before they become critical. If a buyer becomes critical we let her depart unmatched.

\begin{algorithm}[H]
  \caption{The Dynamic Deferred Acceptance algorithm}\label{alg:DDA}
  \begin{itemize}
  \item At any point $t$ during the algorithm, maintain a set of sellers $S_t$, a set of buyers $B_t$, as well as a matching $m$, a price $p_s$ for every seller $s \in S_t$, and a marginal profit $q_b$ for every buyer $b \in B_t$.
  \item Process each event in the following way:
  \begin{enumerate}
    \item\label{step:arrival:seller}\emph{Arrival of a seller s:} Initialize $p_s \leftarrow0$ and $m(s) \leftarrow \emptyset$.
    \item\label{step:arrival:buyer}\emph{Arrival of a buyer $b$:} Start the following {\em ascending auction.} \\
    Repeat
      \begin{enumerate}
         \item Let $q_b \leftarrow \max_{s' \in S_t} v_{s', b} - p_{s'}$ and $s \leftarrow \argmax_{s' \in S_t} v_{s', b} - p_{s'}$.
         \item If $q_b > 0$ then
         \begin{enumerate}
         \item $p_s \leftarrow p_s+\epsilon$.
         \item $m(s)\leftarrow b$ (tentatively match $s$ to $b$)
         \item Set $b$ to $\emptyset$ if $s$ was not matched before. Otherwise, let $b$ be the previous match of $s$.
    	\end{enumerate}
    \end{enumerate}
    Until $q_b \leq 0$ or $b = \emptyset$.
    \item\label{step:departure:seller} \emph{Departure of a seller s:} If seller $s$ becomes critical and  $m(s) \neq \emptyset$, finalize the matching of $s$ and $m(s)$ and collect the reward of $v_{s, m(s)}$.

  \end{enumerate}
\end{itemize}
\end{algorithm}

The ascending auction phase  in our algorithm is similar to the auction algorithm by \cite{bertsekas1988auction}. Prices (for overdemanded sellers) in this auction increase by  $\epsilon$ to ensure termination, and optimality is proven  through $\epsilon$-complementary slackness conditions. For the simplicity of exposition we presented the auction algorithm but for the analysis, we consider the limit $\epsilon \rightarrow 0$ and assume the auction phase terminates with the maximum weight matching. Another way to update the matching is through the  Hungarian algorithm \cite{kuhn1955hungarian}, where  prices are increased simultaneously along an alternating path that only uses edges for which the dual constraint is tight.

The auction phase is always initiated at the existing  prices and profit margins. This, together with the fact that the graph is bipartite, ensures that prices never decrease and and marginal profits never increase throughout the algorithm. Furthermore, the prices and marginal profits of the vertices that are present in the market form an optimum dual for the matching linear program (see Appendix \ref{app:missing_proofs} for more details).
\begin{lemma}
\label{lem:mon}
Consider the DDA algorithm on a constrained bipartite graph.
\begin{enumerate}
\item Throughout the algorithm,  prices corresponding the sellers never decrease and the profit margins of buyers never increase.
\item At the end of every ascending auction,  prices of the sellers and the marginal profits of the buyers form an optimal solution to the dual  of the matching linear program associated with  buyers and sellers present at that particular time.
\end{enumerate}
\end{lemma}

Maintaining a  maximum-weight matching along with optimum dual variables does not guarantee an efficient matching for the whole graph. The dual values are not always feasible for the offline problem. Indeed, the profit margin of some  buyer $b$ may  decrease after some seller departs the market. This is because $b$ may face increasing competition from new buyers, while the bidding process excludes sellers that have already departed the market (whether matched or not).

\begin{proposition}
  \label{prop:factor2}
  DDA is $\nicefrac{1}{2}$-competitive for constrained bipartite  graphs.
\end{proposition}

The proof, given in Section \ref{sec:analysis}, relies on a primal-dual argument. For any arriving buyer $b$ we denote by $q_b^i$ her \emph{initial profit margin} after the ascending auction 
terminates.  When a buyer $b$ is matched or departs, we set $q_{b}^f$ to be her \emph{final profit margin} at that time. Similarly when a seller $s$ departs or is matched, we denote by $p_s^f$ her \emph{final price} at that time.

The proof of the proposition relies on the following three ingredients. First, letting  $S = \cup_{t \in [1,T]} S_t$ and $B = \cup_{t \in [1,T]} B_t$, then the algorithm collects
$$\sum_{s \in S} v_{s, m(s)} = \sum_{s \in S} p_s^f + \sum_{b \in B} q_b^f.$$
Second, although the  final dual variables $(p^f, q^f)$ are not dual feasible, the pair $(p^f, q^i)$ is a feasible dual solution of the offline matching problem.
Finally, we  obtain a  factor $2$ by observing that  $$\sum_{s\in S} p_s^f + \sum_{b \in B} q_b^f = \sum_{b \in B} q_b^i.$$

\subsection{Arbitrary graphs}
\label{sec:arbitrary}

In the previous section, we constructed an algorithm for constrained bipartite graphs, we now extend it to arbitrary graphs.

A naive way to generate a constrained bipartite graph from an arbitrary one is to randomly assign each vertex to be  either a seller or a buyer, independently and with probability \nicefrac{1}{2}. Then we remove any edge between each seller and buyers who have arrived before her. This approach yields the {\it Simple Dynamic Deferred Acceptance} (SDDA) algorithm:

\begin{algorithm}
\caption{Simple Dynamic Deferred Acceptance (SDDA)}
\begin{itemize}
\item For each vertex $t=1,\ldots, T$:
\begin{enumerate}
\item[]  Toss a fair coin to decide whether $t$ is a  \emph{seller} or a \emph{buyer}. Construct the corresponding constrained bipartite graph by keeping only the edges between each buyer and the sellers who arrived at most $d$ steps before her.
\end{enumerate}
\item Run the DDA algorithm on the resulting constrained bipartite graph.
\end{itemize}
\end{algorithm}

\begin{corollary}
  \label{cor:factor8}
  SDDA is $\nicefrac{1}{8}$-competitive for arbitrary graphs.
\end{corollary}

Observe that for $i < j$, edge $(i,j)$ in the original graph remains in the generated bipartite graph with probability $1/4$ (if $i$ is a seller and $j$ is a buyer). We then use proposition \ref{prop:factor2} to prove that SDDA is $\nicefrac{1}{8}$-competitive.

One source of inefficiency in SDDA is that the decision whether an agent is a seller or a buyer is done independently at random and without taking the graph structure into consideration. We next introduce  the   {\it Postponed Dynamic Deferred Acceptance} algorithm that postpones these decisions for as long as possible to enable a more careful construction of the constrained bipartite graph.

When a vertex $k$ arrives, we add two copies of $k$ to an virtual graph: first a seller $s_k$ and then a buyer $b_k$. Seller $s_k$ initially does not have any edges, and buyer $b_k$ has edges towards any vertex $s_l \in S_k$ with value $v_{l,k}$. Then we run the DDA algorithm with the virtual graph as input. When  a vertex $k$ becomes critical, $s_k$ and $b_k$ successively become critical in the virtual graph, and we compute their matches generated by $DDA$.

Both $s_k$ and $b_k$ can be matched in this process.  If we were to honor both matches, the outcome would correspond to a 2-matching, in which each vertex has degree at most 2. Now observe that because of the structure of the constrained bipartite graph, this 2-matching does not have any cycles; it is just a collection of disjoint paths. We decompose each path into two disjoint matchings and choose each matching with probability  $1/2$.

In order to do that, the algorithm must determine, for each original vertex $k$, whether the virtual buyer $b_k$ or virtual seller $s_k$ will be used in the final matching. We will say that $k$ is a \emph{buyer} or \emph{seller} depending on which copy is used. We say that vertex $k$ is \emph{undetermined} when the algorithm has not yet determined which virtual vertex will be used. When an undetermined vertex becomes critical, the algorithm  flips a fair coin to decide whether to match according to the buyer or seller copy. This decision is then propagated to the next vertex in the 2-matching: if $k$ is a \emph{seller} then the next vertex will be a \emph{buyer} and vice-versa.
That ensures that assignments are correlated and saves a factor $2$ compared to uncorrelated assignments in \emph{SDDA}.

\begin{algorithm}[ht]
\caption{Postponed Dynamic Deferred Acceptance (PDDA)}
\label{alg:QDDA}
\begin{itemize}
\item At any point $t$ during the algorithm, maintain an virtual bipartite graph between a set of sellers $S_t$ and a set of buyers $B_t$. Also maintain a matching $m$, a price $p_s$ for every virtual seller $s \in S_t$, and a marginal profit $q_b$ for every virtual buyer $b \in B_t$.
For each (real) vertex $k$, maintain $k$'s status as either \emph{undetermined}, \emph{buyer} or \emph{seller}.

\item Process each event in the following way:
\begin{enumerate}
  \item \emph{Arrival of a vertex $k$:}
  \begin{enumerate}
  	\item Set $k$'s status to be \emph{undetermined}.
    \item \emph{Add a virtual seller:} $S_t \leftarrow S_{t} \cup \{s_k\}$ and $p_{s_k} \leftarrow 0$.
    \item \emph{Add a virtual buyer:} $B_t \leftarrow B_{t} \cup \{b_k\}$ and $q_{b_k} \leftarrow \max_{s \in S_t} v_{s, b_k} - p_s$.
    If $q_{b_k}>0$, start an ascending auction in DDA, and update $m$, $p$, $q$ accordingly.
  \end{enumerate}
  \item \emph{Vertex $k$ becomes critical:}
  \begin{enumerate}
    \item Let $l$ be such that $m(s_k)= b_l$.
Set $B_t \leftarrow B_t \setminus \{b_k\}$, $S_t \leftarrow S_t \setminus \{s_k\}$, and $B_t \leftarrow B_t \setminus \{b_l\}$. (\emph{match in the virtual graph.}) 
    \item If $k$'s status is \emph{undetermined}, w.p $\nicefrac{1}{2}$ set it to be either \emph{seller} or \emph{buyer}.
    \begin{enumerate}
      \item\label{step:seller}\emph{If $k$ is a seller:} finalize the matching of $k$ to $l$ and collect the reward $v_{k,l}$. If $l \neq \emptyset$, set $l$ to be a buyer.
      \item\label{step:buyer}\emph{If $k$ is a buyer:} If $l \neq \emptyset$, set $l$ to be a seller.
    \end{enumerate}
  \end{enumerate}
\end{enumerate}
\end{itemize}
\end{algorithm}

\begin{theorem}
  \label{th:factor4}
  PDDA is $\nicefrac{1}{4}$-competitive.
\end{theorem}

\noindent
The proof of Theorem \ref{th:factor4} is deferred to Section \ref{sec:analysis}. It relies on the following three ingredients.
First, the algorithm collects $$\E\left[\sum_{k \in S} p_{s_k}^f + \sum_{l \in B} q_{b_l}^f\right] = \frac{1}{2} \left(\sum_{k \in [1, T]} p_{s_k}^f + \sum_{l \in [1, T]} q_{b_l}^f\right),$$ where the expectation is taken over the random assignments of \emph{undetermined} vertices to be sellers or buyers.
Second, $(p_{s_k}^f + q_{b_k}^i)_{k \in [1, T]}$ is a feasible dual solution of the offline matching problem. Finally,  similar  to the proof of \ref{prop:factor2}, we use the following equality: $$\sum_{k \in [1, T]} p_{s_k}^f + \sum_{l \in [1, T]} q_{b_l}^f = \sum_{l \in [1, T]} q_{b_l}^i.$$

%% file: analysis.tex

\section{Analysis}
\label{sec:analysis}
In this section, we prove our three main results. We use $\mathcal{A}$ to denote the expected sum of all match values collected by the algorithm, and $\mathcal{O}$ to denote the value of the offline maximum-weight matching.

\subsection{Proof of Proposition \ref{prop:factor2}}

We  prove that \emph{DDA} (Algorithm 1) obtains a competitive ratio of at least $\nicefrac{1}{2}$ on constrained bipartite graphs.  The proof follows the primal-dual framework.

First, we observe that by complementary slackness, any seller $s$ (buyer $b$) that departs unmatched has a final price $p_s^f = 0$ (final profit margin $q_b^f = 0$). When a seller $s$ is critical and matches to $b$, we have  $v_{s,b} = p_s^f + q_b^f$. Therefore, \emph{DDA} collects a reward of $\mathcal{A} = \sum_{s \in S} p_s^f + \sum_{b \in B} q_b^f$.

Second, let us consider a buyer $b$ and a seller $s \in [b - d, b)$ who has arrived before $b$ but not more than $d$ steps before. Because sellers do not finalize their matching before they are critical, we know that $s \in S_b$. An ascending auction may be triggered at the time of $b$'s arrival, after which we have: $v_{s,b} \leq p_s(b) + q_b(b) \leq p_s^f + q_{b}^i$, where the second inequality follows from the definition that $q_b(b) = q_b^i$ and from the monotonicity of sellers' prices (Lemma \ref{lem:mon}). Thus, $(p^f, q^i)$ is a feasible solution to the offline dual problem.

Finally, we observe that upon the arrival of a new buyer, the ascending auction does not change the sum of prices and margins for vertices who were already present:

\begin{claim}
  \label{cl:monotonicity}
  Let $b$ be a new buyer in the market, and let $p, q$ be the prices and margins before $b$ arrived, and let $S_t$ and $B_t$ be the set of sellers and buyers present before $b$ arrived. Let $p'$, $q'$ be the prices and margins at the end of the ascending auction phase (Step 2(a) in Algorithm 1). Then:
\begin{equation}
\sum_{s \in S_t} p_s + \sum_{b \in B_t} q_b = \sum_{s \in S_t} p'_s + \sum_{b \in B_t} q'_b.
\label{eq:conservation}
\end{equation}
\end{claim}
The proof of Claim \ref{cl:monotonicity} is deferred to Appendix \ref{app:missing_proofs}.
By applying this equality iteratively after each arrival, we can relate the initial margins $q^i$ to the final margins $q^f$ and prices $p^f$:
\begin{claim}
  \label{cl:conservation}
  $\sum_{s \in S} p_s^f + \sum_{b \in B} q_b^f = \sum_{b \in B} q_b^i$.
\end{claim}

This completes the proof of Proposition \ref{prop:factor2} given that the offline algorithm achieves at most:
$$\mathcal{O} \leq \sum_{s \in S} p_s^f + \sum_{b \in B} q_b^i \leq 2 \mathcal{A}.$$

It remains to prove Claim \ref{cl:conservation}.

\begin{proof}[Proof of Claim \ref{cl:conservation}]
The idea of the proof is to iteratively apply the result of Claim \ref{cl:monotonicity} after any new arrival.
Let $\widetilde{S}_t$ (resp. $\widetilde{B}_t$) be the set of sellers (buyers) who have departed, or already been matched before time $t$. We show by induction over $t \leq T$ that:

\begin{equation}
  \sum_{s \in \widetilde{S}_t} p_s^f + \sum_{b \in \widetilde{B}_t} q_b^f + \sum_{s \in S_t} p_s(t) +  \sum_{b \in B_t} q_b(t) = \sum_{b \in \widetilde{B}_t} q_b^i +  \sum_{b \in  B_t} q_b^i.
  \label{eq:proof:balance}
\end{equation}
This is obvious for $t = 1$. Suppose that it is true for $t \in [1, T-1]$. Note that departures do not affect \eqref{eq:proof:balance}. If the agent arrivint at $t + 1$ is a seller, then for all other sellers $s$, $p_s(t+1) = p_s(t)$ and for all buyers $b$, $q_b(t+1) = q_b(t)$, thus \eqref{eq:proof:balance}, is clearly still satisfied. Suppose that vertex $t + 1$ is a buyer. Using equation \eqref{eq:conservation}, we have:

$$\sum_{s \in S_{t+1}} p_s(t+1) +  \sum_{b \in B_{t+1}} q_b(t+1) = q_{t+1}(t+1) + \sum_{b \in B_t} q_b(t) + \sum_{s \in S_t} p_s(t) = \sum_{b \in  B_{t+1}} q_b^i.$$

Note that at time $T + d$, every vertex has departed. Thus, $\widetilde{S}_{T+d} = S$, $\widetilde{B}_{T+d} = B$ and $S_{T+d} = B_{T+d} = \emptyset$.
This enables us to conclude our induction and the proof for \eqref{eq:proof:balance}.

\end{proof}

\subsection{Proof of Corollary \ref{cor:factor8}}

We  prove  that \emph{SDDA} is $\nicefrac{1}{8}$-competitive for
arbitrary graphs.

The Offline algorithm solves the following maximum-weight matching problem:

  \begin{equation}
  \begin{split}
  \mathcal{O} = \max & \sum_{k < l \in [1, T]} v_{k, l} x_{k, l} \\
  \text{ s.t. } &\sum_{k < l} x_{k,l} + \sum_{k > l} x_{l, k} \leq 1 \\
   & x_{k,l} \in \{0, 1\}. \\
  \end{split}
  \label{eq:offline:primal}
  \tag{Offline Primal}
  \end{equation}

  Suppose that we have assigned each vertex $k \in [1,T]$ to be either a buyer or a seller with probability $1/2$. For $k < l$, let $\widetilde{v}_{k,l} = v_{k,l} \mathds{1}_{k \in S, l \in B}$.
  Consider the \emph{constrained} offline problem obtained by running \eqref{eq:offline:primal} with edge values $\tilde{v}$. Its expected reward is at least $\nicefrac{1}{4}$ of the offline reward $\mathcal{O}$: Let $x^*$ be an optimal solution to equation \eqref{eq:offline:primal}. It is feasible for the \emph{constrained} problem, and yields value equal to $\mathbb{E}\left[\sum_{k < l \in [1, T]} v_{k,l} x^*_{k,l} \mathds{1}_{k \in S, l \in B} \right] = \frac{1}{4} \mathcal{O}$.

  We can conclude using Claim \ref{cl:conservation} along with the fact that $(p_s^f)_{s \in S}, (q_b^i)_{b \in B}$ is a feasible solution to the constrained offline dual problem.
  This yields the \nicefrac{1}{8} competitive ratio. \qedsymbol

\subsection{Proof of Theorem \ref{th:factor4}}

We   prove that the PDDA algorithm achieves a competitive ratio of at least $\nicefrac{1}{4}$ for  arbitrary graphs.
Observe that because of the randomization, we collect in expectation
$$\mathcal{A} = \E\left[\sum_{t \text{ is a seller}} p_{s_t}^f + \sum_{t \text{ is a buyer}} q_{b_l}^f\right] = \frac{1}{2} \sum_{t \in [1, T]} p_{s_t}^f + q_{b_t}^f.$$

Using the convention that for a pair of vertices $k,l$, $v_{k,l} = 0$ when $|k - l| > d$, the dual of the offline matching problem linear programs can be written as:

\begin{equation}
\begin{split}
\min & \sum_{k \in [1,T]} \lambda_k\\
\text{s.t. } & v_{k,l} \leq \lambda_k + \lambda_l \\ 
& \lambda_k \geq 0.
\end{split}
\label{eq:offline:dual}
\tag{Offline Dual}
\end{equation}

It is enough to show that $(p_{s_k}^f + q_{b_k}^i)_{k \in [1, T]}$ is a feasible solution to equation \eqref{eq:offline:dual}.
This is similar to the proof in the simplified setting. Fix $(k, l) \in [1, T]^2$ and assume that $k < l \leq k + d$. When $l$ arrives, we have $s_k \in S(l)$, therefore $v_{k,l} \leq p_{s_k}(l) + q_{b_l}^i$ by dual feasibility during the online matching procedure. Using the Lemma \ref{lem:mon}, we get feasibility for equation \eqref{eq:offline:dual}:
$$ p_{s_k}^f + q_{b_k}^i + p_{s_l}^f + q_{b_l}^i \geq v_{k,l}.$$

We can conclude the factor $4$ using Claim \ref{cl:conservation}.   \qedsymbol

%% file: stochastic_dep.tex

\section{Extensions: Stochastic departures}

We relax the assumption that all vertices depart after exactly $d$ time steps.
In Section \ref{sec:hardness}, we  show that if  departure times are chosen in an adversarial way, then no algorithm can obtain a constant fraction of the offline matching, even when  departure times are known at the time of arrival.

We focus here on the stochastic case, in which  the departure time $d_i$ of vertex $i$ is sampled independently from a distribution $\mathcal{D}$.
We first assume that the realizations $d_i$ are known upfront (Section \ref{sec:stochastic:known}) and next we consider the case, in which $d_i$ is  revealed only when $i$ becomes critical (Section \ref{sec:stochastic:unknown}).

\subsection{Known departure times}
\label{sec:stochastic:known}

We assume here that the for every agent $i$, her departure time  $d_i$ is sampled i.i.d from a distribution $\mathcal{D}$, and that  $d_i$ is revealed to the online algorithm at the time when $i$ arrives.

\begin{claim}
  Suppose that there exists $\alpha \in (0,1)$ such that $\mathcal{D}$ satisfies the property that for all $i < j$,
  $$\P[i + d_i \leq j + d_j | i + d_i \geq j] \geq \alpha.$$
Then \emph{PDDA} is $\nicefrac{\alpha}{4}-competitive$.
\end{claim}

Observe in particular that if $\mathcal{D}$ is constant with value $d$, we recover our previous result. Furthermore, if $\mathcal{D}$ has a non-decreasing hazard rate, then this property is verified with $\alpha = 1/2$.

\begin{proof}
The main idea is to pre-process the graph by removing edges for which the two endpoints do not arrive and depart in the same order.
Consider a modified graph, where we set edge value $v_{i,j}$ to $0$ when $i + d_i > j + d_j$. Each non-zero edge is kept with probability at least $\alpha$.

Note that the offline optimal matching $x^*$ on the initial graph is a feasible matching on the modified graph. Thus, the offline matching on the modified graph collects a reward of at least $\sum_{i < j \in [1, T]} x^*_{i,j} v_{i,j} \mathbb{I}_{i + d_i \leq j + d_j} \geq \alpha \mathcal{O}$.

Observe that the PDDA algorithm only requires that when a buyer becomes critical, any compatible seller has already departed. This is the case in our modified graph, which yields our factor $\nicefrac{\alpha}{4}$.
\end{proof}

\subsection{Unknown departure times}
\label{sec:stochastic:unknown}
We assume now that the online algorithm only learns $d_i$ once $i$ becomes critical.
The main difficulty is that a buyer $b$ may become critical before the seller $s = m(b)$ that she is tentatively matched to. Because we want the seller to wait until she becomes critical, we cannot conduct the match. However, the departure of $b$ may cause the price $p_s$ to decrease, which violates the monotonicity property in Lemma \ref{lem:mon}.

We  modify the PDDA algorithm in the following way: we set the buyer copy $b_k$ to never become critical in the auxiliary graph. Because buyer vertices can only match to sellers who arrive before them, there will eventually be a time when $b_k$ does not have any edge left in the auxiliary graph.

\begin{proposition}
  \label{prop:single_bid}
  Suppose that there exists $\alpha \in (0,1)$ such that $\mathcal{D}$ satisfies the property that for all $i < j$,
  $$\P[i + d_i \leq j + d_j | i + d_i \geq j] \geq \alpha.$$
  Then PDDA is $\nicefrac{\alpha}{4}$-competitive.
\end{proposition}

\begin{proof}
When a vertex $k$ becomes critical in the original graph, if $\rho(k) = S$, we try to match vertex $k$ to  vertex $l$ such that $m(s_k) = b_l$. With probability at least $\alpha$, vertex $l$ is still present in the original graph.
\end{proof}

\begin{corollary}
  PDDA is $\nicefrac{1}{8}$-competitive when $\mathcal{D}$ has a non-decreasing hazard rate.
\end{corollary}

%% file: hardness.tex

\section{Examples}
\label{sec:hardness}

We will present six examples. The first one shows an upper bound of $\nicefrac{1}{2}$ for the online matching on arbitrary graphs. The second and third show that, even in the case of a bipartite \emph{constrained} graph, randomized and deterministic algorithms cannot obtain competitive ratios higher than $0.8$ and $0.618$ respectively. The fourth one shows that our analyses of Proposition \ref{prop:factor2} and Theorem \ref{th:factor4} are tight. The last two examples show that no algorithm is constant-competitive in the case where we let departures be chosen by an adversary, or if the departures are stochastic and the algorithm does not know when vertices become critical.

\subsection*{Upper Bounds}
\begin{claim}
No deterministic algorithm is constant-competitive, and no randomized algorithm is more than $\nicefrac{1}{2}$-competitive.
\end{claim}
\begin{proof}
	Observe that in Figure \ref{fig:hard:basic},
\end{proof}

\begin{claim}
  \label{cl:ex:bipartite_constrained}
  When the input is a bipartite \emph{constrained} graph:
  \begin{itemize}
    \item[-] No deterministic algorithm can obtain a competitive ratio above $\frac{\sqrt{5} - 1}{2} \approx 0.618$.
    \item[-] No randomized algorithm can obtain a competitive ratio above $\frac{4}{5}$.
  \end{itemize}
\end{claim}

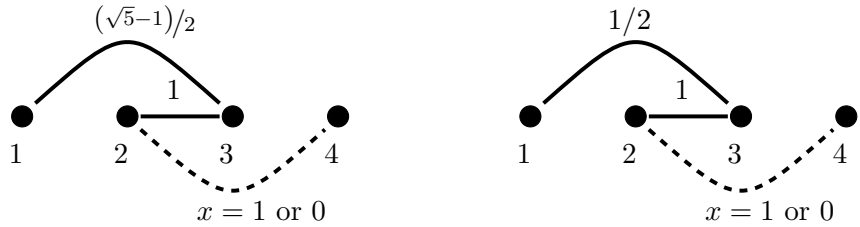
\begin{figure}[!ht]
  \centering
\begin{tikzpicture}[pre/.style={<-,shorten <=1.5pt,>=stealth,thick}, post/.style={->,shorten >=1pt,>=stealth,thick}, scale=0.7]
\tikzstyle{every node}=[draw,shape=rectangle,minimum size=5mm, inner sep=0];
\tikzstyle{edge} = [draw,thick,-]
\tikzstyle{every node}=[shape=circle,minimum size=8mm, inner sep=0];

\draw [fill](-12,0) circle [radius=0.2];
\node [right] at (-12.7,-0.7) {$1$};
\draw [fill](-10,0) circle [radius=0.2];
\node [right] at (-10.7,-0.7) {$2$};
\draw [fill](-8,0) circle [radius=0.2];
\node [right] at (-8.7,-0.7) {$3$};
\draw [fill](-6,0) circle [radius=0.2];
\node [right] at (-6.7,-0.7) {$4$};

\draw[line width=1.5pt] [-] (-10+0.25, 0) .. controls(-9,0) ..(-8-0.25, 0);
\node [right] at (-9.7,0.5) {$1$};
\draw[line width=1.5pt] [-] (-12+0.25, 0.25) .. controls(-10,1.8) ..(-8-0.25, 0.25);
\node [right] at (-10.7,1.8) {$\nicefrac{\left(\sqrt{5} - 1\right)}{2}$};
\draw[dashed, line width=1.5pt] [-] (-10+0.25, -0.25) .. controls(-8,-1.8) .. (-6-0.25, -0.25);
\node [right] at (-8.7,-1.8) {$x = 1$ or $0$};

\end{tikzpicture}
\hspace{1.5cm}
\begin{tikzpicture}[pre/.style={<-,shorten <=1.5pt,>=stealth,thick}, post/.style={->,shorten >=1pt,>=stealth,thick}, scale=0.7]
\tikzstyle{every node}=[draw,shape=rectangle,minimum size=5mm, inner sep=0];
\tikzstyle{edge} = [draw,thick,-]
\tikzstyle{every node}=[shape=circle,minimum size=8mm, inner sep=0];

\draw [fill](-12,0) circle [radius=0.2];
\node [right] at (-12.7,-0.7) {$1$};
\draw [fill](-10,0) circle [radius=0.2];
\node [right] at (-10.7,-0.7) {$2$};
\draw [fill](-8,0) circle [radius=0.2];
\node [right] at (-8.7,-0.7) {$3$};
\draw [fill](-6,0) circle [radius=0.2];
\node [right] at (-6.7,-0.7) {$4$};

\draw[line width=1.5pt] [-] (-10+0.25, 0) .. controls(-9,0) ..(-8-0.25, 0);
\node [right] at (-9.7,0.5) {$1$};
\draw[line width=1.5pt] [-] (-12+0.25, 0.25) .. controls(-10,1.8) ..(-8-0.25, 0.25);
\node [right] at (-10.7,1.8) {$1/2$};
\draw[dashed, line width=1.5pt] [-] (-10+0.25, -0.25) .. controls(-8,-1.8) .. (-6-0.25, -0.25);
\node [right] at (-8.7,-1.8) {$x = 1$ or $0$};
\end{tikzpicture}

\caption{Bipartite graph where $S = \{1, 2\}$ and $B = \{3, 4\}$, with $d = 2$: vertex $1$ becomes critical before $4$ arrives. The adversary is allowed to choose edge $(2,4)$ to be either  $1$ or $0$. Left: instance for the deterministic case. Right: instance for the randomized case.}
\label{fig:ex:constrained}
\end{figure}

\begin{proof}
  {\bf Deterministic case:} Consider the example on the left of Figure \ref{fig:ex:constrained}. When seller $1$ becomes critical, the algorithm either matches her to buyer $3$, or lets $1$ depart unmatched. The adversary then chooses $x$ accordingly. Thus the competitive ratio cannot exceed:
  $$ \max \left(\min_{x \in \{0, 1\}} \frac{\frac{\sqrt{5} - 1}{2} + x}{\max(\frac{\sqrt{5} - 1}{2} + x, 1)}, \min_{x \in \{0, 1\}} \frac{1}{\max(\frac{\sqrt{5} - 1}{2} + x, 1)} \right) = \frac{\sqrt{5} - 1}{2}.$$

  {\bf Stochastic case:} Consider the example on the right of Figure \ref{fig:ex:constrained}. Similarly to the deterministic case, when seller $1$ becomes critical, the algorithm decides to match her to $3$ with probability $p$. The adversary then chooses $x$ accordingly. Thus the competitive ratio cannot exceed:
  $$ \max_{p \in [0,1]} \min_{x \in \{0, 1\}} \frac{p(1/2 + x) + (1-p)}{\max(1/2 + x, 1)} = 4/5.$$
\end{proof}

\subsection*{Tightness of the analysis}
We will now show that our analyses for both the DDA and PDDA algorithms are tight:
\begin{claim}
  \label{cl:tightness}
  There exists a \emph{constrained} bipartite graph for which DDA is $\nicefrac{1}{(2 - \epsilon)}$-competitive and for which the PDDA is $\nicefrac{1}{(4 - 2 \epsilon)}$ -competitive.
\end{claim}

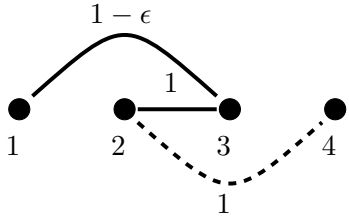
\begin{figure}[!ht]
  \centering

\def\radius{2.6}
\def \Pointsize {1.4pt}
\begin{tikzpicture}[pre/.style={<-,shorten <=1.5pt,>=stealth,thick}, post/.style={->,shorten >=1pt,>=stealth,thick}, scale=0.7]
\tikzstyle{every node}=[draw,shape=rectangle,minimum size=5mm, inner sep=0];
\tikzstyle{edge} = [draw,thick,-]
\tikzstyle{every node}=[shape=circle,minimum size=8mm, inner sep=0];

\draw [fill](-12,0) circle [radius=0.2];
\node [right] at (-12.7,-0.7) {$1$};
\draw [fill](-10,0) circle [radius=0.2];
\node [right] at (-10.7,-0.7) {$2$};
\draw [fill](-8,0) circle [radius=0.2];
\node [right] at (-8.7,-0.7) {$3$};
\draw [fill](-6,0) circle [radius=0.2];
\node [right] at (-6.7,-0.7) {$4$};

\draw[line width=1.5pt] [-] (-10+0.25, 0) .. controls(-9,0) ..(-8-0.25, 0);
\node [right] at (-9.7,0.5) {$1$};
\draw[line width=1.5pt] [-] (-12+0.25, 0.25) .. controls(-10,1.8) ..(-8-0.25, 0.25);
\node [right] at (-10.7,1.8) {$1 - \epsilon$};
\draw[dashed, line width=1.5pt] [-] (-10+0.25, -0.25) .. controls(-8,-1.8) .. (-6-0.25, -0.25);
\node [right] at (-8.7,-1.8) {$1$};
\end{tikzpicture}

\caption{Bipartite graph where $S = \{1, 2\}$ and $B = \{3, 4\}$, with $d = 2$: vertex $1$ becomes critical before $4$ arrives. Dotted edges represent edges that are not know to the algorithm initially.}
\label{fig:ex:tightness}

\end{figure}

\begin{proof}
Consider the input graph in Figure \ref{fig:ex:tightness}.

{\bf DDA case:} Vertex $2$ will be temporarily matched to $3$, and vertex $1$ will depart unmatched, hence the factor $\nicefrac{1}{2}$.

{\bf PDDA case:} Similarly, $1$ will depart unmatched. When $2$ becomes critical, with probability $\nicefrac{1}{2}$, she will be determined to be a \emph{buyer} and will depart unmatched. Therefore the PDDA collects in expectation $\nicefrac{1}{2}$ while the offline algorithm collects $2 - \epsilon$. 
\end{proof}

\subsection*{Relaxing our assumptions}
We consider the \emph{Adversarial departures} (AD) setting, where the adversary is allowed to choose the departure time $d_i$ of vertex $i$. We assume that the online algorithm knows $d_i$ at the time of arrival of $i$.

\begin{claim}
No algorithm is constant-competitive in the AD setting.
\end{claim}
\begin{proof}
  Let us consider a graph $G_K$ with $n$ vertices, where $K$ will be chosen later by the adversary. For all $j \in [2,K]$, $v_{1,j} = M^j$, and for all other $(i,j)$, $v_{i,j} = 0$. Assume that vertex $1$ departs after $n$ arrivals, while all other vertices depart right away. For $k \geq 2$, let $x_k$ be the probability that the online algorithm matches vertex $1$ to $k$ when $k$ arrives.
  Observe that because the algorithm does not know $K$, the $x$ has to be valid when $K = n$. Therefore, $\sum_{k = 1}^n x_k \leq 1$.
  Therefore, there exists $k$ such that $x_k \leq 1/n$. The adversary chooses $K = k$. This implies that $\A = \sum_{j \leq k} x_j M^j \leq \frac{M^k}{n} + M^{k-1} \leq \frac{2}{n} \mathcal{O}$, where the last inequality is obtained by taking $M \geq n$.
\end{proof}

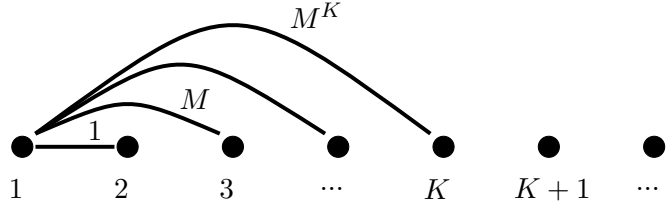
\begin{figure}[!ht]
  \centering
  \begin{tikzpicture}[pre/.style={<-,shorten <=1.5pt,>=stealth,thick}, post/.style={->,shorten >=1pt,>=stealth,thick}, scale=0.7]
    \tikzstyle{every node}=[draw,shape=rectangle,minimum size=5mm, inner sep=0];
    \tikzstyle{edge} = [draw,thick,-]
    \tikzstyle{every node}=[shape=circle,minimum size=8mm, inner sep=0];

    \draw [fill](-14,0) circle [radius=0.2];
    \node [right] at (-14.7,-0.8) {$1$};
    \draw [fill](-12,0) circle [radius=0.2];
    \node [right] at (-12.7,-0.8) {$2$};
    \node [right] at (-13.2, 0.3) {$1$};
    \draw[line width=1.5pt] [-] (-14+0.25, 0) .. controls(-13,0) ..(-12-0.25, 0);

    \draw [fill](-10,0) circle [radius=0.2];
    \node [right] at (-10.7,-0.8) {$3$};
    \node [right] at (-11.3,0.9) {$M$};
    \draw[line width=1.5pt] [-] (-14+0.25, 0.25) .. controls(-12,1) ..(-10-0.25, 0.25);

    \draw [fill](-8,0) circle [radius=0.2];
    \node [right] at (-8.7,-0.8) {$...$};
    \draw[line width=1.5pt] [-] (-14+0.25, 0.25) .. controls(-11,2) ..(-8-0.25, 0.25);

    \draw [fill](-6,0) circle [radius=0.2];
    \node [right] at (-6.7,-0.8) {$K$};
    \node [right] at (-9,2.5) {$M^K$};
    \draw[line width=1.5pt] [-] (-14+0.25, 0.25) .. controls(-10,3) ..(-6-0.25, 0.25);

    \draw [fill](-4,0) circle [radius=0.2];
    \node [right] at (-4.7,-0.8) {$K+1$};
    \draw [fill](-2,0) circle [radius=0.2];
    \node [right] at (-2.7,-0.8) {$...$};

  \end{tikzpicture}
    \caption{Graph where $d_1 = N$ and $d_i = 0$ for all $i > 1$. Vertices $k > K$ have no edges.}
  \label{fig:adv_dep}
\end{figure}

We now consider the \emph{Adversarial departure distribution} (ADD) setting where $d_i$ are sampled i.i.d from a distribution chosen by the adversary. We assume furthermore that the online algorithm knows the realization $d_i$ upon the arrival of vertex $i$.

\begin{claim}
No algorithm is constant-competitive in the ADD setting
\end{claim}
\begin{proof}
  The idea is that we can construct a graph that exhibits the same properties as Figure \ref{fig:adv_dep}, even with i.i.d departures. Fix $n \geq 1$ and assume that departures are distributed according to the following distribution: w.p. $\nicefrac{1}{n}$, $d_i = n^2$, and w.p. $(1 - \nicefrac{1}{n})$, $d_i = 0$.

  The arrivals are defined as follows: the first $n \ln{n}$ vertices have no edge between themselves. The adversary then selects $K \in [0, \sqrt{n}]$.
  For $i \leq n \ln{n}$ and $j \in [n \ln{n}, n \ln{n} + K]$, $v_{i,j} = M^j$. Where $M$ is a large constant to be defined later. Vertices $j > n \ln{n} + K$ have no edges.

  Let $E_1$ be the event that there exists $j \in [n \ln{n} + 1, n\ln{n}+\sqrt{n}]$ such that $d_j > 1$.
  $$\P[E_1] = 1 - (1 - 1/n)^{\sqrt{n}} = \frac{1}{\sqrt{n}} + o\left(\nicefrac{1}{\sqrt{n}}\right).$$

  Let $E_2$ be the event that there are less than $\sqrt{n}$ vertices $j \in [0, n\ln{n}]$ such that $d_j > 1$. The (random) number of such vertices $X$ is binomially distributed with parameters $n \ln{n}$ and $\nicefrac{1}{n}$. Thus
  $$\P[E_2] = \P[X \leq \sqrt{n}] \geq 1 - O(\nicefrac{1}{\sqrt{n}}).$$

  We can write:
  $\E[\A] \leq \P[E_1 \cup E_2^c]\sum_{k = 0}^K M^k + \E[\A \mid E_1^c \cap E_2]$
  Conditional on $E_1^c$ and $E_2$, the best any algorithm can do is match each vertex $j \in [n\ln{n}+1, n\ln{n}+ \sqrt{n}]$ with probability $\ln{n} / \sqrt{n}$.
  Therefore $\E[\A] \leq O(\nicefrac{\ln{n}}{\sqrt{n}}) \sum_{k = 0}^K M^k$.

  Let $E_3$ be the event that there exists $i \in [0, n\ln{n}]$ such that $d_i > 1$. $\P[E_3] \geq (1 - \frac{1}{n})$. Therefore, $\E[\mathcal{O}] \geq M^k \P[E_3] \geq (1 - \frac{1}{n}) M^k$.
  For M large enough, we can conclude that $\nicefrac{\E[\A]}{\E[\mathcal{O}]} = O(\nicefrac{1}{\sqrt{n}})$.

\end{proof}

Finally, we consider the \emph{Stochastic Unknown Departures} (SUD) setting where the adversary chooses a departure rate $\delta$, and departures $d_i$ are i.i.d geometric random variables with parameter $\delta$. The online  algorithm knows $\delta$ but does not know $d_i$ even when $i$ becomes critical.

\begin{claim}
  Even when the departure process $\mathcal{D}$ is memoryless, if the algorithm doesn't know when vertices become critical, it cannot obtain a constant competitive ratio.
\end{claim}

\begin{proof}
  Consider the graph in Figure \ref{fig:adv_dep}, and assume that at each time step, an unmatched vertex in the graph has probability $1 - \delta$ of departing.

  Conditional on vertex $1$ being present at time $k$, let $x_k$ be the probability that the algorithm matches $1$ to $k$. We have: $\A = \sum_{k = 2}^K x_k \delta^k M^k$.
  Furthermore, observe that $\sum_{k = 2}^n x_k \leq 1$. Therefore there exists $l$ such that $x_l \leq 1/n$. Then take $K = l$.
  $$\A \leq x_{l} \delta^{l} M^{l} + \delta^{l + 1} M^l + \sum_{k \leq l - 1} \delta^{k} M^{k} \leq \mathcal{O} \left(x_{l} + \delta + \sum_{k
  \leq l - 1} \delta^{k} M^{k - l}\right).$$
  Therefore  $\nicefrac{\A}{\mathcal{O}} \leq O(\nicefrac{1}{n})$ for $\delta$ small enough and $M$ large enough.
\end{proof}

%% file: numerical.tex

\section{Numerical results}
In \emph{PDDA}, we compute a 2-matching over the vertices that are currently present, and select each edge with probability $\nicefrac{1}{2}$. Although this randomization is useful to hedge against the worst case instance, it may be ineffective when the compatibility graph is not generated by an adversary.

Nonetheless, some of the key ideas behind \emph{PDDA} may be useful to construct algorithms that perform well on non-adversarial graphs. This motivates a modified version of the Dynamic Defered Dcceptance algorithm, termed \emph{MDDA}, in which vertices are no longer separated into buyers and sellers and every vertex now has a \emph{price}.

Under MDDA, prices are reset to $0$ after each arrival, and  an auction is conducted on the non-bipartite graph: while there exists an unmatched vertex, one is selected  at random and it bids on its most prefered match. This leads to a tentative matching $M$ over all the vertices that are currently present in the graph\footnote{Note that because the graph is non-bipartite, this auction mechanism may fail to converge to a maximum-weight matching.}.  When a vertex becomes critical, it is matched  according to the tentative matching $M$. This modified algorithm does not have theoretical guarantees, but we will show that it performs well on data-driven compatibility graphs.

Note that MDDA can be thought of as a re-optimization algorithm, where the optimal matching is re-computed when new information becomes available. One can replace the auction mechanism with any algorithm that computes a maximum-weight matching. We implemented this algorithm, termed here \emph{Re-Opt}, where the matching is found by solving a mixed-integer program at each time step.

We  compare the MDDA and Re-Opt  algorithms against three benchmarks that have been previously consdiered, or that are commonly used in practice:
\begin{itemize}
  \item[-] The \emph{Greedy} algorithm. The algorithm  matches vertices as soon as possible to their available neighbor with the highest value (ties are broken in favor of the earliest arrival).
  \item[-] The \emph{Batching}($k$) algorithm. The algorithm waits $k$ times-steps  and then finds a maximum-weight matching. Unmatched vertices are kept in the next batch.\footnote{See \cite{agatz2011dynamic,ashlagi2013kidney} in the case of ride-sharing and kidney exchange respectively.} We  report the best simulation results across parameters $k = 5, 10, 50, 100, 200, 300$.
  \item[-] The Patient algorithm. This algorithm waits until waits until a vertex becomes critical, and matches it to the neighbor with the highest  match value (ties are broken in favor of the earliest arrival). This allows to seperate the value from  knowing the time in which vertices become  critical and the value of optimization.
\end{itemize}

\subsection*{Data}

The first experiment uses data from the National Kidney Registry (NKR), which consists of 1681 patient-donor pairs who have joined the NKR. For any two patient-donor pairs $k$ and $t$, we can determine whether the patient from each pair would have been medically eligible to receive the kidney from the other pair's donor, had they had been present at the same time. If that is the case, we set $v_{k,l} = 1$, otherwise  and $v_{k,l} =0$ (in particular we simply try to maximize the number of matches)\footnote{We ignore here the possibility of larger cyclic exchanges (3 or 4-cycles) or chains, which are common in practice.}.

In the second instance, we use New York City yellow cabs dataset \footnote{http://www.andresmh.com/nyctaxitrips/}, which consists of rides taken in NYC over a year. For any pair $k,l$ of trips, we can compute the Euclidian distance that would have been traveled had the two passengers taken the same taxi (with multiple stops). The value $v_{k,l}$ represents the ``distance saved'' by combining the trips.

In both cases, this enables us to generate a dynamic graph in the following way. For $t \in [1,T]$:
\begin{enumerate}
 \item Sample with replacement an arrival $t$ from the dataset.
 \item For any vertex $l$ that is present in the pool, compute the value $v_{t,l}$ of matching $t$ to $l$.
 \item Sample a departure time $d_t \sim \mathcal{D}$.
\end{enumerate}

We  consider two settings, termed \emph{deterministic} and \emph{stochastic} respectively, in which $\mathcal{D}$ is either constant with value $d$, or exponentially distributed with mean $d$. We will report simulation results for $d = 50, 100, 200, 300$.

\subsection*{Results}

In Figure \ref{fig:taxi}, we observe that both the \emph{Patient} and \emph{Batching} algorithms outperform \emph{Greedy}. Intuitively, having vertices wait until they become critical helps to thicken the market and gives vertices higher valued match options. We notice that when departures are deterministic, \emph{Batching} with the optimal batch size will be almost as efficient as \emph{Re-Opt}. However when the departures are stochastic, there is value in matching vertices as they become critical (\emph{MDDA} and \emph{Re-Opt}).

\begin{figure}[ht!]
  \centering
\includegraphics[scale=0.34]{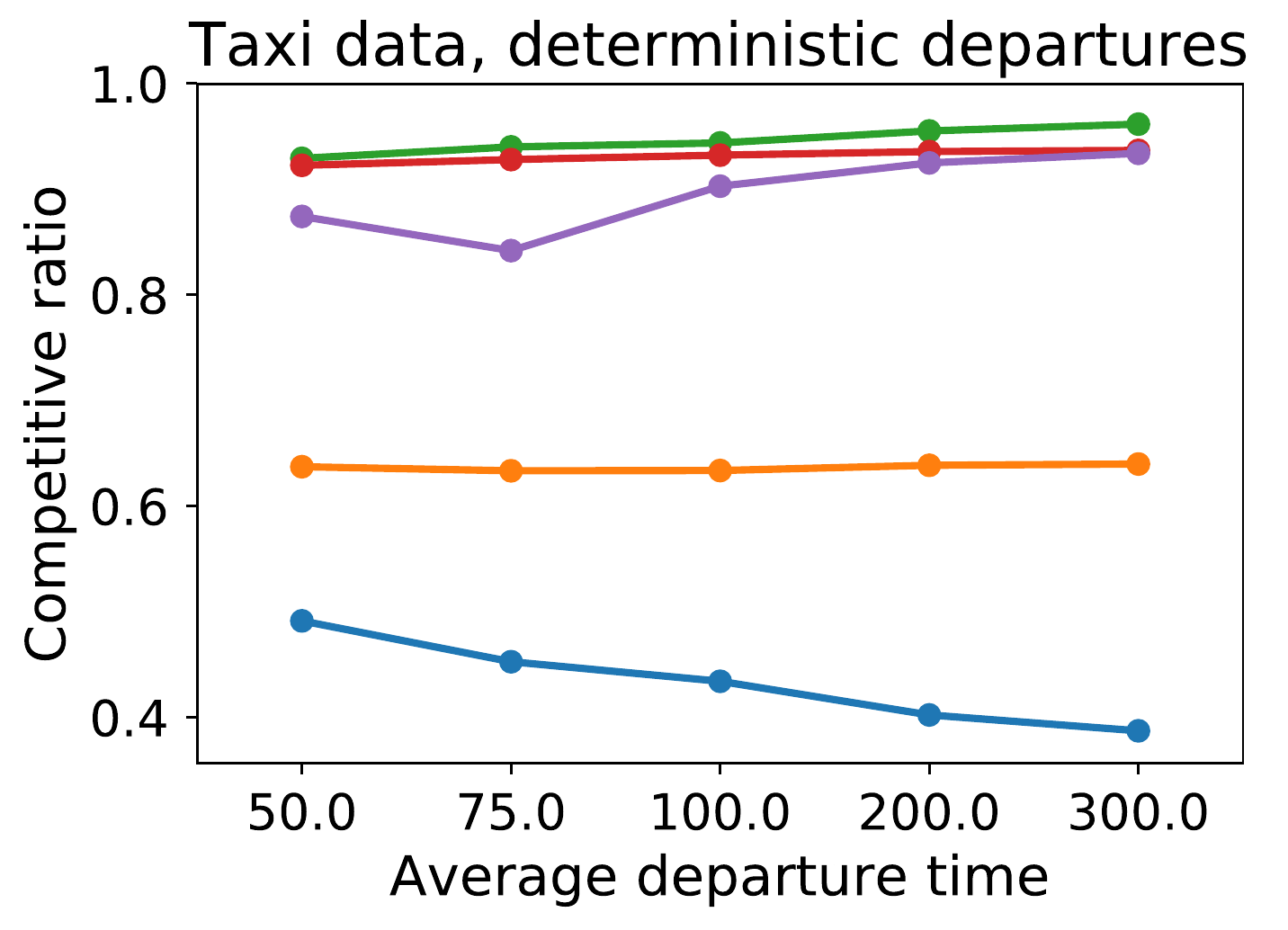}
\hspace{0.2cm}
\includegraphics[scale=0.34]{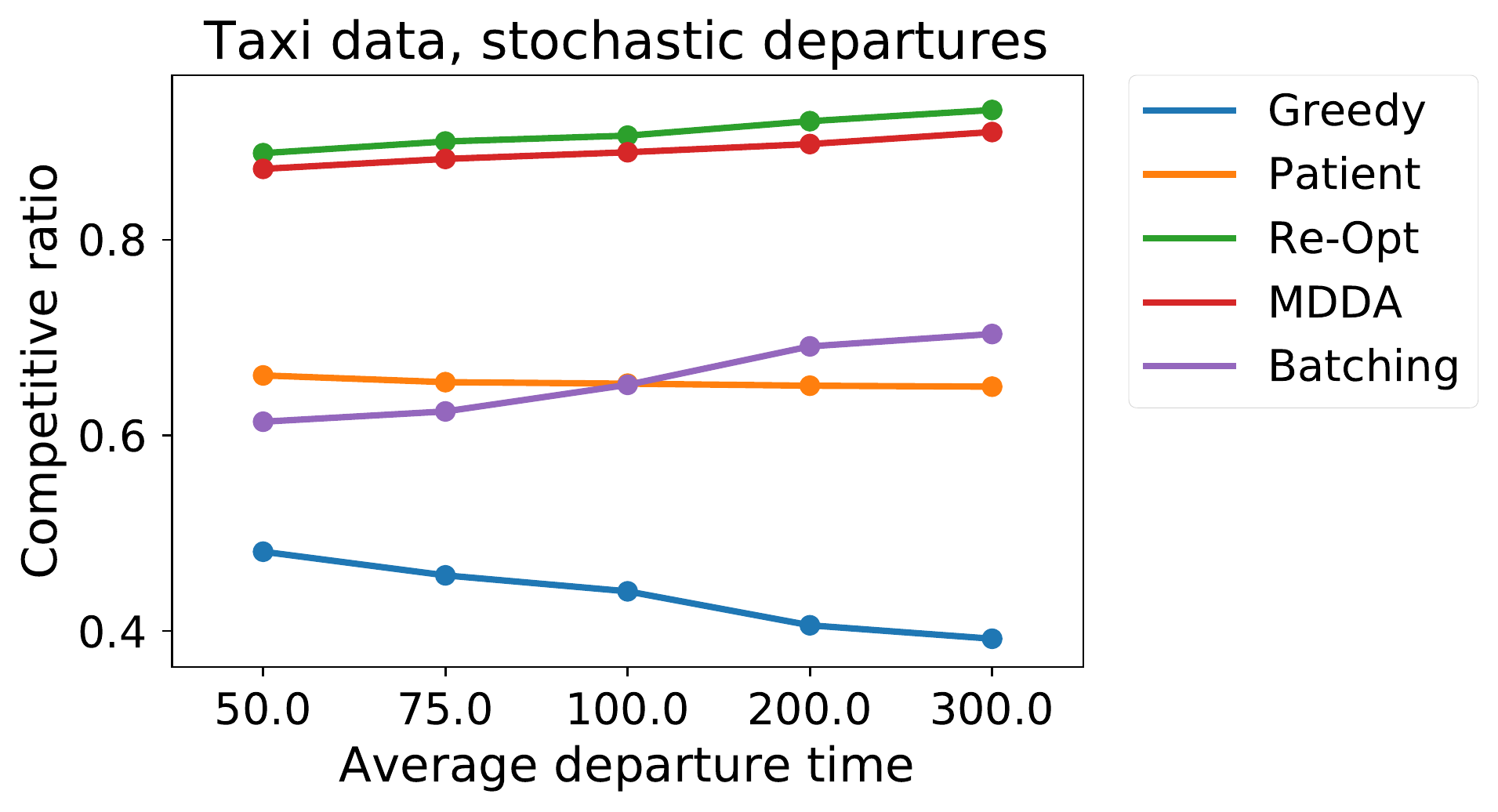}
\caption{Performance of our 4 algorithms on taxi data (weighted compatibility graph).}
\label{fig:taxi}
\end{figure}

Figure \ref{fig:kidney} provides the results for the kidney exchange simulations. In this case, the compatibility graph is unweighted, which implies that because we break ties in favor of the earliest arrival and departures are in order of arrival, \emph{Greedy} and \emph{Patient} are equivalent. Again, \emph{Batching} performs relatively well in the deterministic case (if $b<d$,  no vertex will depart before it is included in one batch), but very poorly in the stochastic case. 

\begin{figure}[ht!]
  \centering
\includegraphics[scale=0.33]{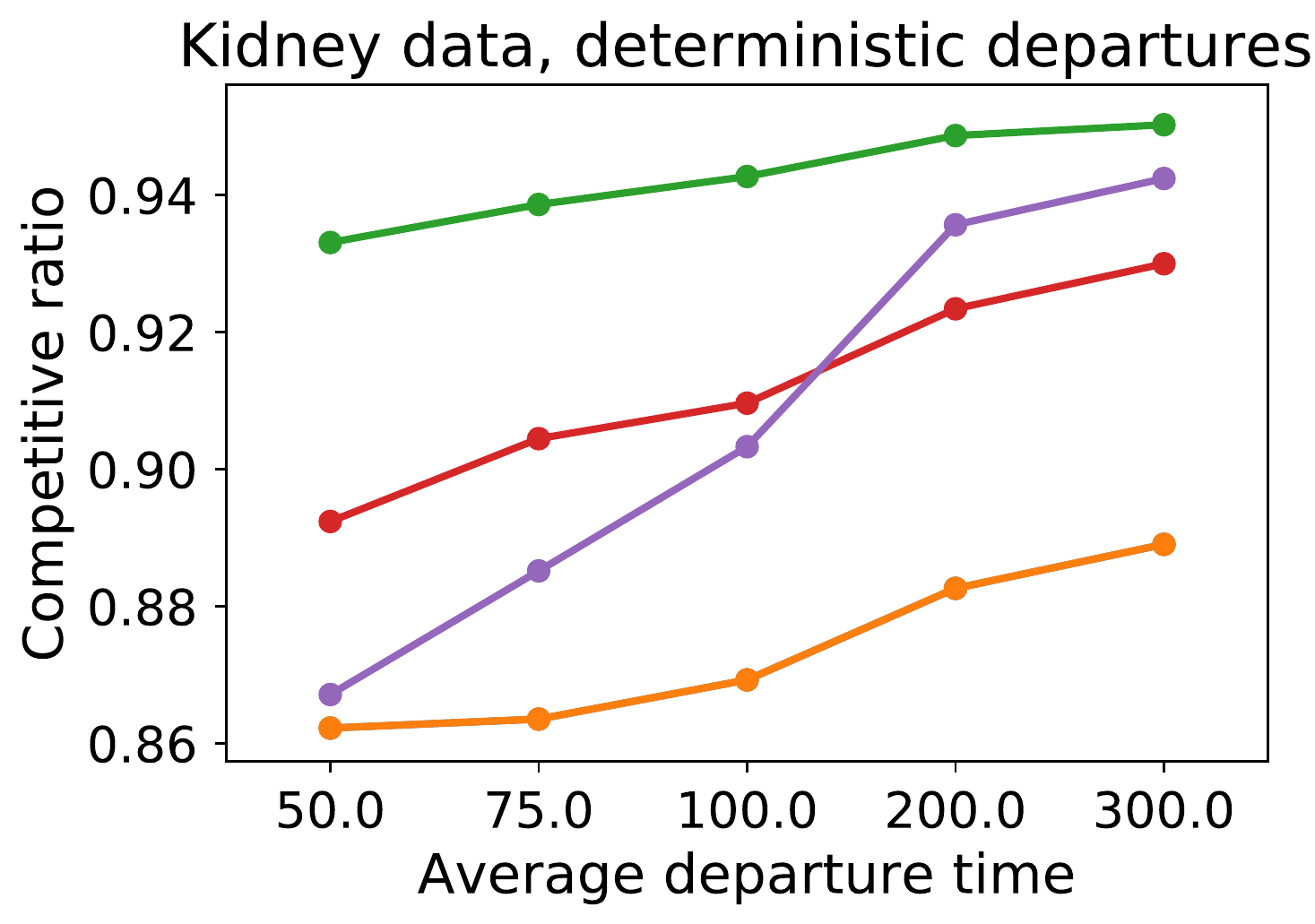}
\hspace{0.2cm}
\includegraphics[scale=0.33]{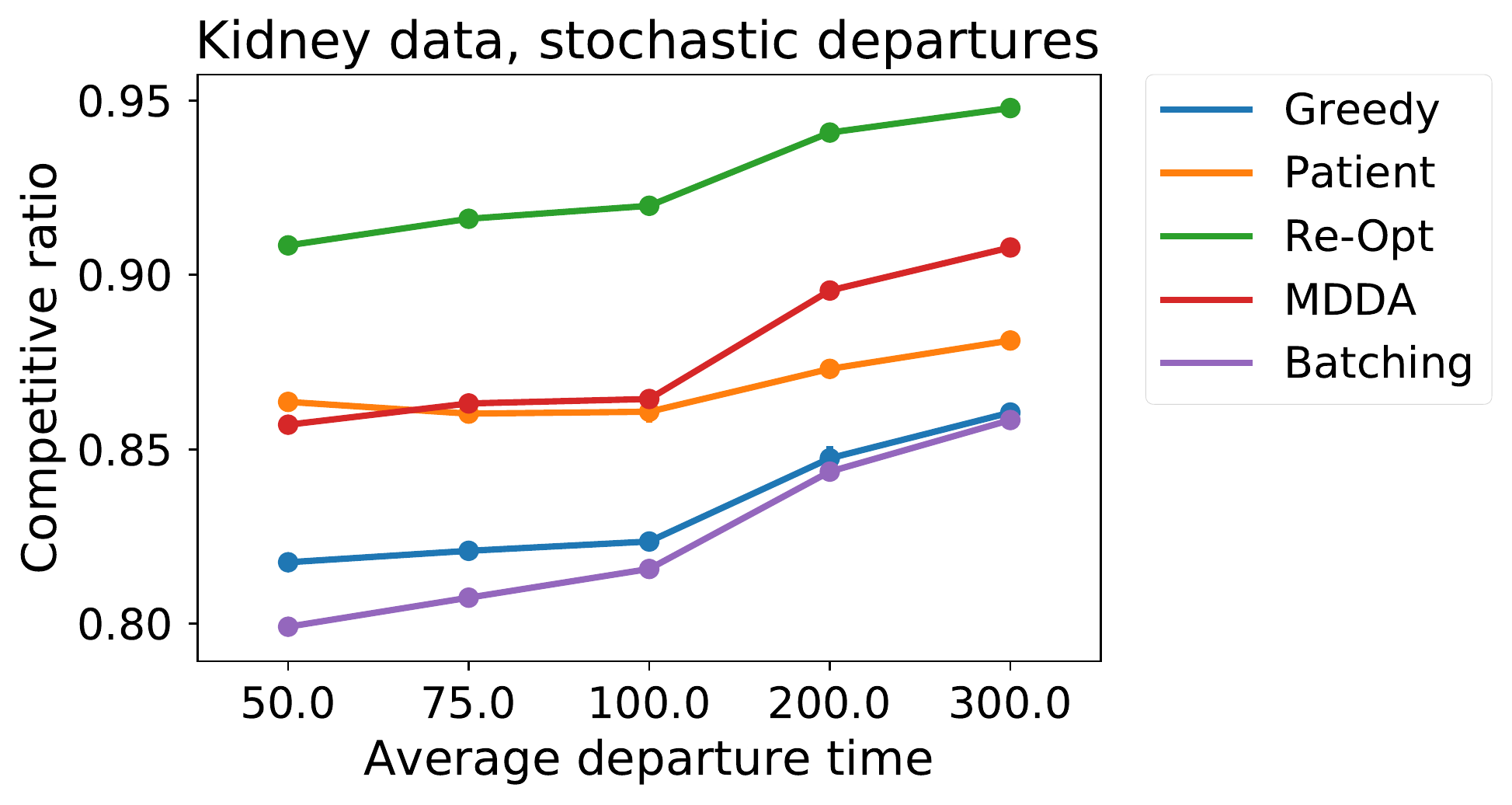}
\caption{Performance of our 4 algorithms on kidney data (unweighted compatibility graph).}
\label{fig:kidney}
\end{figure}

In both datasets, we observe that \emph{Re-Opt} outperforms all other algorithm, although in the cases where the departures are deterministic, \emph{Batching} performs close to \emph{Re-Opt} when the batch size $k$ is carefully chosen.  This shows the value of both having information on agents' departure times and also  subsequent optimization.

It is important to note that the experiments we ran do not take into account the cost of waiting. We think that a richer model that accounts for this would be an interesting future direction.
Two interesting areas for future work include the setting when the information about agent's departure times is uncertain, as well as models that are less restrictive than the adversarial setting (see, e.g., \citet{ozkan2016dynamic}).

%% file: conclusion.tex

\section{Conclusion}

This  paper introduces a model for dynamic matching, in which all agents arrive and depart over time. Match values are heterogenous and the underlying graph is arbitrary and can thus can be non-bipartite. We study algorithms that perform well across any sequence of arrivals and on any set of match values. 


Importantly,  our model imposes restrictions on the departure process and requires the  algorithm to know when vertices become critical. There are many interesting directions for future research. An immediate open problem is to close the gap between the upper bound of $\nicefrac{1}{2}$, and the achievable competitive ratios ($\nicefrac{1}{4}$ for deterministic departures, and $\nicefrac{1}{8}$ for stochastic departures). Independently, our model imposes that matches retain the same value regardless of when they are conducted. Being able to account for agent's waiting costs would also be very interesting.
Another direction is to design algorithms that achieve both a high total value but also a large  fraction of matched agents. 


%
%

%% file: appendix.tex
\section{Missing proofs}
\label{app:missing_proofs}
\subsection{Missing proof of Claim \ref{cl:monotonicity}}

\begin{proof}
  The proof of termination in \cite{bertsekas1988auction} relies on the introduction of a minimum bid $\epsilon$ in step $6$ of the auction algorithm to ensure that the algorithm does not get stuck in a cycle of bids of $0$. In the limit where $\epsilon \rightarrow 0$, the algorithm ressembles the \emph{hungarian algorithm} \cite{kuhn1955hungarian}. The idea is to search for an augmenting path along the edges for which the dual constraint is tight. If such a path is found, the matching is augmented, otherwise we perform simultaneous bid increases in way that ensures that prices $p$ and margins $q$ are still dual feasible.

  We assume that we are given at time $t$ an optimal matching $m$ and optimal duals $(p, q)$ corresponding to the graph with vertices $S_t, B_t$. We assume that we added a new vertex $b^*$ to $B'_{t} = B_t \cup \{b^*\}$, and that we initialized $q_{b^*} = \max_{s \in S_t} v_{s, b^*} - p_s$ 

  Initialize $m' = m$, $p' = p$, $q' = q$. Note that primal and dual feasibility are satisfied. Therefore, $(m', p', q')$ is optimal iff the following three complementary slackness condition are satisfied:
\begin{equation}
  \label{eq:CS:dual}
  \forall s \in S_t, v_{s,m(s)} = p'_s + q'_m(s).
  \tag{CS1}
\end{equation}
\begin{equation}
  \label{eq:CS:primal}
  \forall s \in S'_t, m(s) = \emptyset \implies p_s = 0.
  \tag{CS2}
\end{equation}
\begin{equation}
  \label{eq:CS:primal2}
  \forall b \in B'_t, m(b) = \emptyset \implies q_b = 0.
  \tag{CS3}
\end{equation}
  Note that \eqref{eq:CS:dual} and \eqref{eq:CS:primal} are already satisfied. If $q'_{b^*} = 0$ then \eqref{eq:CS:primal2} is also satisfied and we have an optimal solution.

  Suppose now that $q'_{b^*} > 0$. We will update $(m', p', q')$ in a way that maintains primal and dual feasibility, as well as \eqref{eq:CS:dual} and \eqref{eq:CS:primal}.

  Our objective is to find an augmenting path in the graph. First we will start by trying to find an alternating path that starts on $b$ and only uses edges for which the dual constraint is tight: $\mathcal{E} = \{(s,b) | s \in S'_t, b \in B'_t,  v_{s,b} = p'_s + q'_b\}$. Observe that by \eqref{eq:CS:dual} all the matched edges in $m$ are in $\mathcal{E}$. We will now successively color vertices as follows:
  \begin{itemize}
    \item[0.] Start by coloring $b^*$ in blue.
    \item[1.] For any blue buyer $b$, for any seller $s$ such that $(s, b) \in \mathcal{E}$ and $s \neq m(b)$, we color $s$ in red.
    \item[2.] For any red seller $s$, let $b = m(s)$, then color $b$ in blue.
  \end{itemize}

  Observe that there is an alternating path between $b^*$ and any red seller.
  If at one point we color an unmatched seller $s^*$ in red, this means that we have found an augmenting path from $b^*$ to $s^*$ that only utilizes edges in $\mathcal{E}$. In that case, we change $m'$ according to the augmenting path. Because of the way we chose edges in $\mathcal{E}$, \eqref{eq:CS:dual} is still satisfied. \eqref{eq:CS:primal} and \eqref{eq:CS:primal2} are now also satisfied, which means we have an optimal solution $(m', p',q')$.

  We terminate when we are unable to color vertices any further. In that case, let us define $\delta_1 = \min_{b \text{ blue}} q_b$. If $\delta_1 = 0$, then there exists $b \in B_{t}'$ with $q_b = 0$ and an alternating path form $b^*$ to $b$. We update $m'$ according to that path, and verify that all CS conditions are now satisfied.

  Suppose that $\delta_1 > 0$. Define
  \begin{equation}
    \delta_2 = \min_{b \text{ blue, } s \text{ not red}} \{ p_s + q_b - v_{s,b} \}.
    \label{eq:delta2}
  \end{equation}

  The fact that we cannot color any more vertices implies that $\delta_2 > 0$.
  Let $\delta = \min(\delta_1, \delta_2) > 0$. For every red seller $s$, we update $p'_s \leftarrow p'_s + \delta$. For every blue buyer $b$, we update $q'_b \leftarrow q'_b - \delta$. Observe that dual feasibility is still verified, as well as \eqref{eq:CS:dual}.

  If $\delta = \delta_2$, taking $(s,b)$ the argmin in \eqref{eq:delta2}, we now have such that $p'_s + q'_b - v_{s,b} = 0$ which means we can add $(s,b)$ to $\mathcal{E}$ and color $s$ in red.
  We will eventually have $\delta = \delta_1$, and this leads to $q_b = 0$ and we can terminate.
  This proves both the termination and correctness. Furthermore, monotonicity of the dual variables is also straightforward. Let us now prove the conservation property:
  \begin{equation}
    \label{eq:proof:conservation}
    \sum_{s \in S_t} p_s + \sum_{b \in B_t} q_b = \sum_{s \in S_t} p'_s + \sum_{b \in B} q'_b.
  \end{equation}
  Note that when we update the dual variables, then every seller we colored in red was matched in $S'$ and we colored that match in blue. Therefore, apart from the initial vertex $i$, there are the same number of red and blue vertices.
\end{proof}

%% file: adv_paper.bbl
\begin{thebibliography}{}

\bibitem[Agatz et~al., 2011]{agatz2011dynamic}
Agatz, N.~A., Erera, A.~L., Savelsbergh, M.~W., and Wang, X. (2011).
\newblock Dynamic ride-sharing: A simulation study in metro atlanta.
\newblock {\em Transportation Research Part B: Methodological},
  45(9):1450--1464.

\bibitem[Akbarpour et~al., 2017]{akbarpour2017thickness}
Akbarpour, M., Li, S., and Oveis~Gharan, S. (2017).
\newblock Thickness and information in dynamic matching markets.

\bibitem[Anderson et~al., 2015]{anderson2015dynamic}
Anderson, R., Ashlagi, I., Gamarnik, D., and Kanoria, Y. (2015).
\newblock A dynamic model of barter exchange.
\newblock In {\em Proceedings of the twenty-sixth annual ACM-SIAM symposium on
  Discrete algorithms}, pages 1925--1933. Society for Industrial and Applied
  Mathematics.

\bibitem[Ashlagi et~al., 2017a]{ashlagi2017min}
Ashlagi, I., Azar, Y., Charikar, M., Chiplunkar, A., Geri, O., Kaplan, H.,
  Makhijani, R., Wang, Y., and Wattenhofer, R. (2017a).
\newblock Min-cost bipartite perfect matching with delays.
\newblock In {\em LIPIcs-Leibniz International Proceedings in Informatics},
  volume~81. Schloss Dagstuhl-Leibniz-Zentrum fuer Informatik.

\bibitem[Ashlagi et~al., 2017b]{ashlagi2017matching}
Ashlagi, I., Burq, M., Jaillet, P., and Manshadi, V. (2017b).
\newblock On matching and thickness in heterogeneous dynamic markets.

\bibitem[Ashlagi et~al., 2013]{ashlagi2013kidney}
Ashlagi, I., Jaillet, P., and Manshadi, V. (2013).
\newblock Kidney exchange in dynamic sparse heterogenous pools.
\newblock {\em arXiv preprint arXiv:1301.3509}.

\bibitem[Baccara et~al., 2015]{baccara2015optimal}
Baccara, M., Lee, S., and Yariv, L. (2015).
\newblock Optimal dynamic matching.
\newblock Working paper.

\bibitem[Bertsekas, 1988]{bertsekas1988auction}
Bertsekas, D.~P. (1988).
\newblock The auction algorithm: A distributed relaxation method for the
  assignment problem.
\newblock {\em Annals of operations research}, 14(1):105--123.

\bibitem[Demange et~al., 1986]{demange1986multi}
Demange, G., Gale, D., and Sotomayor, M. (1986).
\newblock Multi-item auctions.
\newblock {\em Journal of Political Economy}, 94(4):863--872.

\bibitem[Dickerson et~al., 2013]{dickerson2013failure}
Dickerson, J.~P., Procaccia, A.~D., and Sandholm, T. (2013).
\newblock Failure-aware kidney exchange.
\newblock In {\em Proceedings of the fourteenth ACM conference on Electronic
  commerce}, pages 323--340. ACM.

\bibitem[Dutta et~al., 2017]{dutta2017}
Dutta, C., Greenhall, A., Puranmalka, K., and Sholley, C. (2017).
\newblock Online matching in a ride sharing platform.

\bibitem[Emek et~al., 2016]{emek2016online}
Emek, Y., Kutten, S., and Wattenhofer, R. (2016).
\newblock Online matching: haste makes waste!
\newblock In {\em Proceedings of the forty-eighth annual ACM symposium on
  Theory of Computing}, pages 333--344. ACM.

\bibitem[Feldman et~al., 2009]{aryanak_stmatching}
Feldman, J., Mehta, A., Mirrokni, V.~S., and Muthukrishnan, S. (2009).
\newblock Online stochastic matching: Beating 1-1/e.
\newblock In {\em Proceedings of the 50th Annual IEEE Symposium on Foundations
  of Computer Science (FOCS)}, pages 117--126.

\bibitem[Gale and Shapley, 1962]{gale1962college}
Gale, D. and Shapley, L.~S. (1962).
\newblock College admissions and the stability of marriage.
\newblock {\em The American Mathematical Monthly}, 69(1):9--15.

\bibitem[Goel and Mehta, 2008]{aryanak_randominput}
Goel, G. and Mehta, A. (2008).
\newblock Online budgeted matching in random input models with applications to
  adwords.
\newblock In {\em Proceedings of the nineteenth annual ACM-SIAM symposium on
  Discrete algorithms (SODA)}, pages 982--991.

\bibitem[Hu and Zhou, 2016]{hu2016dynamic}
Hu, M. and Zhou, Y. (2016).
\newblock Dynamic type matching.

\bibitem[Huang et~al., 2018]{huang2018}
Huang, Z., Kang, N., Tang, Z.~G., Wu, X., and Zhang, Y. (2018).
\newblock How to match when all vertices arrive online.

\bibitem[Jaillet and Lu, 2013]{STMatchingPatrick}
Jaillet, P. and Lu, X. (2013).
\newblock Online stochastic matching: New algorithms with better bounds.
\newblock {\em Mathematics of Operations Research}, 39(3):624--646.

\bibitem[Karp et~al., 1990]{kvv}
Karp, R.~M., Vazirani, U.~V., and Vazirani, V.~V. (1990).
\newblock An optimal algorithm for on-line bipartite matching.
\newblock In {\em Proceedings of the twenty-second annual ACM symposium on
  Theory of computing (STOC)}, pages 352--358.

\bibitem[Kuhn, 1955]{kuhn1955hungarian}
Kuhn, H.~W. (1955).
\newblock The hungarian method for the assignment problem.
\newblock {\em Naval Research Logistics (NRL)}, 2(1-2):83--97.

\bibitem[Manshadi et~al., 2011]{mos}
Manshadi, V.~H., Oveis-Gharan, S., and Saberi, A. (2011).
\newblock Online stochastic matching: online actions based on offline
  statistics.
\newblock In {\em Proceedings of the Twenty-Second Annual ACM-SIAM Symposium on
  Discrete Algorithms (SODA)}, pages 1285--1294.

\bibitem[Mehta, 2013]{mehta2013online}
Mehta, A. (2013).
\newblock Online matching and ad allocation.
\newblock {\em Foundations and Trends{\textregistered} in Theoretical Computer
  Science}, 8(4):265--368.

\bibitem[Mehta et~al., 2007]{mehta2007adwords}
Mehta, A., Saberi, A., Vazirani, U., and Vazirani, V. (2007).
\newblock Adwords and generalized online matching.
\newblock {\em Journal of the ACM (JACM)}, 54(5):22.

\bibitem[Ozkan and Ward, 2016]{ozkan2016dynamic}
Ozkan, E. and Ward, A.~R. (2016).
\newblock Dynamic matching for real-time ridesharing.

\bibitem[Truong and Wang, 2018]{truong2018}
Truong, V.-A. and Wang, X. (2018).
\newblock Online matching in a ride sharing platform.

\bibitem[{\"U}nver, 2010]{Utku}
{\"U}nver, M.~U. (2010).
\newblock {Dynamic Kidney Exchange}.
\newblock {\em Review of Economic Studies}, 77(1):372--414.

\end{thebibliography}
